\documentclass[letterpaper, 10 pt, conference]{ieeeconf}  

\IEEEoverridecommandlockouts                              

\RequirePackage{xcolor}

\usepackage{cite}
\usepackage{amsmath,amssymb,amsfonts}

\usepackage{multirow}
\usepackage{hhline}

\usepackage{url}
\usepackage{tikz}
\usepackage{pgfplots}
\usepackage{grffile}
\usepackage{stfloats}
\pgfplotsset{compat=newest}
\usepgfplotslibrary{groupplots}   
\usetikzlibrary{plotmarks}
\usetikzlibrary{arrows.meta}
\usepgfplotslibrary{patchplots}
\usepackage{siunitx}
\usetikzlibrary{shapes,arrows}
\tikzstyle{block} = [draw, fill=white, rectangle, minimum height=3em, minimum width=4em]
\tikzstyle{sum} = [draw, fill=white, circle, node distance=1cm]
\usepackage{todonotes}
\usepackage{graphicx}
\usepackage{diagbox}
\usepackage{textcomp}
\let\labelindent\relax
\usepackage{enumitem}
\usepackage{siunitx}
\usepackage{booktabs}
\usepackage[hidelinks]{hyperref}
\usepackage[noabbrev]{cleveref}

\usepackage{mathtools}
\usepackage{mleftright}
\mleftright

\usepackage{algorithm}
\usepackage[noend]{algpseudocode}
\algrenewcommand\algorithmicindent{1em}
\renewcommand{\algorithmicrequire}{\textbf{Input:}}

\newtheorem{lem}{Lemma}
\newtheorem{thm}{Theorem}
\newtheorem{cor}{Corollary}

\newcommand{\bbR}{\mathbb{R}}
\newcommand{\bbN}{\mathbb{N}}
\newcommand{\bbZ}{\mathbb{Z}}
\newcommand{\bbC}{\mathbb{C}}

\newcommand{\calO}{\mathcal{O}}

\newcommand{\calC}{\mathcal{C}}

\newcommand{\calN}{\mathcal{N}}

\newcommand{\bfc}{\mathbf{c}}

\newcommand{\bfg}{\mathbf{g}}
\newcommand{\bfU}{\mathbf{U}}
\newcommand{\bfs}{\mathbf{s}}

\newcommand{\sfB}{\mathsf{B}}

\newcommand{\Mult}{\mathsf{Mult}}

\newcommand{\RotCt}{\mathsf{RotCt}}
\newcommand{\RotVec}{\mathsf{RotVec}}
\newcommand{\Rot}{\mathsf{Rot}}

\newcommand{\KeyGen}{\mathsf{KeyGen}}
\newcommand{\Enc}{\mathsf{Enc}}
\newcommand{\Dec}{\mathsf{Dec}}

\newcommand{\sk}{\mathsf{sk}}

\newcommand{\pt}{\mathsf{pt}}
\newcommand{\ckks}{\mathsf{ct}}
\newcommand{\poly}{\mathsf{poly}}

\def\BibTeX{{\rm B\kern-.05em{\sc i\kern-.025em b}\kern-.08em
    T\kern-.1667em\lower.7ex\hbox{E}\kern-.125emX}}
\begin{document}

\title{\LARGE \bf Variational Encrypted Model Predictive Control}

\author{Jihoon Suh, Yeongjun Jang, Junsoo Kim, and Takashi Tanaka
\thanks{*This work was supported in part by DARPA COMPASS (HR0011-25-3-0210; in part by AFOSR DSCT (FA9550-25-1-0347); and in part by the National Research Foundation of Korea(NRF) grant funded by the Korea government (MSIT) (No. RS-2024-00353032)
}
\thanks{J.~Suh, and T.~Tanaka are with the School of Aeronautics and Astronautics, Purdue University, West Lafayette, IN 47907, USA (e-mail: \{suh95, tanaka16\}@purdue.edu).}
\thanks{Y.~Jang is with ASRI, Department of Electrical and Computer Engineering, Seoul National University, Seoul, 08826, Korea (email: jangyj0512@snu.ac.kr).
}
\thanks{J.~Kim is with the Department of Electrical and Information Engineering,
Seoul National University of Science and Technology, Seoul, 01811, Korea (email: junsookim@seoultech.ac.kr).
}
}

\maketitle
\thispagestyle{empty}
\pagestyle{empty}

\begin{abstract}
We develop a variational encrypted model predictive control (VEMPC) protocol whose online execution relies only on encrypted polynomial operations.
The proposed approach reformulates the MPC problem into a sampling-based estimator, in which the computation of the quadratic cost is naturally handled by tilting the sampling distribution, thus reducing online encrypted computation.
The resulting protocol requires no additional communication rounds or intermediate decryption, and scales efficiently through two complementary levels of parallelism.
We analyze the effect of encryption-induced errors on optimality, and simulation results demonstrate the practical applicability of the proposed method.
\end{abstract}

\section{INTRODUCTION}

Model Predictive Control (MPC) is widely adopted in industrial control systems as it enables explicit constraint handling while providing theoretical guarantees \cite{MuskRawl93}.
Since MPC requires solving an optimization problem at every time step, outsourcing this computation to a cloud server is attractive.
However, transmitting system states, model parameters, or trajectory data to an external server raises privacy concerns.
These data are vulnerable to eavesdropping in transit and to inference by a semi-honest server that executes the prescribed protocol while attempting to learn private information.

Homomorphic encryption (HE) is considered as a potential solution to mitigate these concerns by enabling computation directly on encrypted data, thereby protecting both transmission and computation. However, HE only supports addition and multiplication, whereas solving the underlying MPC typically requires non-polynomial operations (e.g., projection or comparison) to enforce constraints.
For this reason, projected gradient based results have either required intermediate decryption and projection by the client at each iteration \cite{AlexMora18}, or performed only a single gradient step to avoid additional communication rounds, leading to suboptimal control performance \cite{FranPuig24}.
Another line of work precomputes explicit feedback laws offline and evaluates them using HE \cite{DaruRedd18,SchlDaru20,PengXiew26}.
However, this requires identifying the operating region of the current state online through comparison operations, which is delegated to the sensor \cite{DaruRedd18,PengXiew26} or implemented via a two-party protocol \cite{SchlDaru20} involving additional communication.
More recently, \cite{SchlIann25} reformulated MPC as a sequence of penalized unconstrained problems that can be solved using polynomial operations; while it removes excessive communication rounds, its real-time control suitability remains uncertain.

In this paper, we propose a variational approach to encrypted MPC that converts the constrained MPC optimization into a sampling-based estimation problem.
To circumvent the high computational cost of evaluating the quadratic cost term over encrypted data, we \textit{tilt} the reference sampling distribution, allowing it to naturally absorb this cost. 
Consequently, the proposed protocol requires only a single polynomial evaluation per control execution step in the encrypted domain.
The resulting encrypted MPC architecture requires a one-time offline exchange and a single round of communication between the client and cloud at each online execution step, without any intermediate decryption.
Moreover, the cloud side computation can leverage two different levels of parallelism, enabling efficient encrypted MPC execution.
While the variational approach has been applied, for example, in probabilistic inference \cite{HoffBlei13,WainJord08} and stochastic optimal control \cite{Kapp05,Todo06}, to the best of our knowledge, this is the first work to integrate it with encrypted MPC.

\subsection{Main Contributions}
\begin{itemize}[leftmargin=*]
    \item We propose variational encrypted MPC (VEMPC), in which the quadratic cost term is absorbed into a closed-form tilted Gaussian distribution, enabling efficient online execution without projection or iterative optimization.
    
    \item We show that both the sampled control trajectory and its constraint residual are affine functions of the sampling noise, enabling an HE-compatible protocol based on encrypted addition and low-degree polynomial evaluation.

    \item We utilize two complementary levels of parallelism in the VEMPC protocol, one at the plaintext sample-level and the other at the ciphertext-level. Together, these enable efficient online execution within tens of milliseconds (see Table~\ref{tab:ablation}), closing the gap towards real-time implementation.

    \item We provide open-source implementations in Python and Go using OpenFHE-Python \cite{OpenFHE} and Lattigo \cite{lattigo}, two widely adopted and efficient HE libraries.
\end{itemize}

\subsection{Notation}
The sets of real, natural, integers and complex numbers are denoted by $\bbR$, $\bbN$, $\bbZ$ and $\bbC$, respectively. 
For a vector $x \in \bbC^n$, $\|x\|$ denotes the infinity norm.
The Hadamard product (element-wise multiplication) is denoted by $\odot$.
We denote by $\RotVec_\rho(x)$ the cyclic upward shift of $x$ by $\rho$ positions (e.g., $\RotVec_1([1, 2, 3, 4]^\top) = [2, 3, 4, 1]^\top$).
For scalars or vectors $v_1,\ldots,v_n$ , $[v_1;\cdots;v_n] := [v_1^\top, \ldots, v_n^\top]^\top$. Boldface symbols denote encrypted quantities.
Lastly, $[\,\cdot\,]_+ := \max\{\,\cdot\,,0\}$.

\section{Preliminaries: CKKS Cryptosystem}\label{prelim}
We briefly review the Cheon-Kim-Kim-Song (CKKS) cryptosystem \cite{cheon2017homomorphic}, a fully HE scheme that enables approximate arithmetic directly on encrypted data.

The CKKS scheme includes algorithms $\KeyGen$, $\Enc$, and $\Dec$. 
Given a security parameter $\lambda_{\rm sec}\in\bbN$, the key generation algorithm $\KeyGen$ outputs a secret key $\sk$. 
For a plaintext $\mathsf{pt} \in\bbC^{N_{\ckks} /2}$, where $N_{\ckks} \in\bbN$ is a power-of-two, the encryption algorithm takes $\sk$ and outputs a ciphertext $\Enc(\mathsf{pt})\in\calC$ with $\calC$ denoting the ciphertext space. 
The decryption algorithm uses $\sk$ to approximately recover the underlying plaintext, as $\Dec(\Enc(\mathsf{pt}))\approx \mathsf{pt}$.
The approximation error arises from scaling and rounding during encoding and depends primarily on the scaling factor $\Delta>0$ and the security parameter $N_{\ckks}$ (ring dimension).

The primitive homomorphic operations like addition, multiplication, and rotation are represented respectively by
\begin{align*}
    \oplus:\calC\times \calC \to \calC,~~~~\otimes:\calC\times \calC \to \calC,~~~~\RotCt_\rho:\calC\to\calC.
\end{align*}

The following lemma summarizes the growth of errors induced by the approximate homomorphic operations.
For conciseness, the bounds are stated using $\calO$-notation and concrete derivations can be found in \cite{cheon2017homomorphic, cheon2018bootstrapping}.

\begin{lem}\label{lem:homo}
    There exist $\sfB^{\Enc}, \sfB^{\Mult},\sfB^{\Rot} \in \calO(N_\ckks/\Delta)$ for any $\mathsf{pt} \in \bbC^{N_\ckks/2}$, $\bfc, \bfc' \in \calC$ and $\rho\in\bbZ$, such that
        \begin{align*}
            &\left\|\Dec(\Enc(\mathsf{pt})) - \mathsf{pt} \right\| \le \sfB^{\Enc}, \\
            &\left\|\Dec(\bfc \oplus \bfc') - \left(\Dec(\bfc) + \Dec(\bfc')\right) \right\| = 0, \label{eq:homoAdd} \\
            &\left\|\Dec(\Enc(\pt) \otimes \Enc(\pt')) - \pt \odot \pt'\right\| \\
            &\le \!\left(\left\| \pt \right\| \!+\! \left\| \pt' \right\| \right) \sfB^\Enc \!+\! (\sfB^\Enc)^2 \!+\! \sfB^{\Mult}=: \sfB^\otimes(\|\pt\|,\|\pt'\|), 
            \\
            &\left\| \Dec(\RotCt_\rho(\bfc)) - \RotVec_\rho(\Dec(\bfc)) \right\|
            \le \sfB^\Rot.
        \end{align*}
\end{lem}

By sequentially composing the operations from Lemma \ref{lem:homo}, one can homomorphically evaluate an arbitrary polynomial \cite{cheon2017homomorphic, cheon2018bootstrapping}. 
With a slight abuse of notation, let $h_\ell(\bfc) \in \calC$ denote the homomorphic evaluation of a polynomial $h_\ell$ on a ciphertext $\bfc\in\calC$. 
The following corollary provides a bound on the resulting error growth.

\begin{cor}\label{cor:poly}
Consider a polynomial\footnote{When applied to a vector, $h_\ell$ acts in a componentwise manner.} $h_\ell(\gamma) = \sum_{k=0}^{\ell} c_k \gamma^k$ of order $\ell\in\bbN$, where $c_k\in\bbC$ for $k=0,\ldots,\ell$.
For any $\bfc\in\calC$,
\begin{multline*}
    \left\| \Dec(h_\ell(\bfc)) - h_\ell(\Dec(\bfc)) \right\| \\
    \le \ell\cdot\sfB^\Mult\max\left(1, \|\Dec(\bfc)\|\right)^{\ell-1}  \textstyle\sum_{k=0}^\ell |c_k| =:\sfB^{\poly}(h_{\ell},\bfc).
\end{multline*}
\end{cor}

The parameters $N_{\ckks}$ and $\Delta$ affect not only the error growths in Lemma~\ref{lem:homo} and Corollary~\ref{cor:poly} but also the security level of the scheme. 
Typically, increasing $N_{\ckks}$ or decreasing $\Delta$ strengthens security, potentially at the cost of higher computational cost. 
In this letter, we assume that these parameters are given based on standard security requirements, and focus on integrating the CKKS scheme into our control framework. 
For detailed parameter selection methods, see \cite{cheon2017homomorphic}.

\section{Problem Formulation}
We consider solving a finite-horizon linear–quadratic model predictive control (LQ-MPC) problem in a client–cloud architecture. The objective is to utilize the cloud's computational resources to synthesize the control input, while ensuring that the client's sensitive information such as current state or model parameters remains private to the cloud via encryption.

Consider the discrete-time linear time-invariant system
\begin{equation}\label{eq:sys}
x(t+1) = A x(t) + B u(t),
\end{equation}
where $x(t) \in \mathbb{R}^n$ and $u(t) \in \mathbb{R}^m$ denote the system state and control input at time step $t\in\bbZ$, respectively.

At each MPC update step $t$, the measured state $x(t)$ serves as the initial state of the prediction horizon, denoted by $x_0 := x(t)$. 
The controller plans a sequence of control inputs $U := [u_0; u_1; \cdots; u_{N-1}] \in \mathbb{R}^{Nm}$ over a prediction horizon of length $N\in\bbN$ by solving the following optimization:
\begin{align}\label{eq:lqmpc}
    &\min_{U} J_0(U;x_0) := x_N^\top Q_f x_N + 
    \textstyle\sum_{k=0}^{N-1} \left( x_k^\top Q x_k + u_k^\top R u_k \right),\nonumber \\
    &~~\mbox{s.t.}~GU\le h(x_0), 
\end{align}
where $x_k$ is the predicted state corresponding to $U$ driven by \eqref{eq:sys}, and $Q_f,Q \succeq 0$ and $R \succ 0$ are the weight matrices.
The inequality $G U \le h(x_0)$, with $G \in \mathbb{R}^{p \times Nm}$ and $h(x_0) \in \mathbb{R}^{p}$, compactly represents state and input constraints, where $p\in\bbN$ denotes the total number of inequalities.

From the optimizer, only the first optimal input $u_0$ is applied to the system, and the process repeats at the next time step with the updated state.

We can rewrite the standard LQ-MPC problem \eqref{eq:lqmpc} in a compact quadratic program (QP) form\cite[Ch.~8]{borrelli2017predictive}.
The dynamics \eqref{eq:sys} induces the following prediction model
\begin{equation}
X = \Lambda x_0 + \Psi U,
\end{equation}
where $X = [x_1;\cdots;x_N]\in\bbR^{Nn}$ and the matrices $\Lambda\in\bbR^{Nn\times n}$ and $\Psi\in\bbR^{Nn\times Nm}$ are determined by the system parameters $(A,B)$.
Substituting this model into the cost function yields the compact quadratic cost
\begin{equation}
J_0(U;x_0) = \frac12 U^{\top} H U +  x_0^\top S U + x_0^\top P x_0,
\end{equation}
where the matrices $H\in\bbR^{Nm\times Nm}$, $S\in\bbR^{n\times Nm}$, and $P\in\bbR^{n\times n}$ are determined by $(A,B,Q,R,Q_f)$.

By defining the constraint residual $g(U;x_0) := G U - h(x_0)$, the feasible set can be characterized as 
\begin{equation}\label{eq:feasible_set}
\mathcal{F}(x_0) = \{ U \in \mathbb{R}^{Nm} : g(U;x_0) \le 0 \}.
\end{equation}
Dropping the constant term $x_0^\top P x_0$ yields the equivalent QP:
\begin{equation}\label{eq:mpc_qp_compact}
\min_{U \in \mathcal{F}(x_0)} \quad J_Q(U;x_0) := \frac12 U^\top H U + x_0^\top S U .
\end{equation}

\section{A Variational Approach to Encrypted MPC}
Directly adopting standard methods for solving QP can be difficult because iterative and branching-based algorithms are not well suited for HE arithmetic.
Instead, we employ a variational approach that leads to a HE-compatible, sampling-based algorithm to estimate the solution of the QP \eqref{eq:mpc_qp_compact}. 

\subsection{Variational Reformulation}\label{subsec:variational}
Define the extended-value penalty on the feasible set~\eqref{eq:feasible_set}:
\begin{equation}\label{eq:phi}
\phi(U;x_0) := 
    \begin{cases}
    0, & U\in \mathcal{F}(x_0),\\
    +\infty, & \text{otherwise}.
    \end{cases}
\end{equation}
Using this penalty, the compact QP~\eqref{eq:mpc_qp_compact} can be expressed as an unconstrained problem via the extended value functional
\begin{equation}\label{eq:hard_unconstrained}
    \min_{U \in \mathbb{R}^{Nm}} \mathcal{J}(U;x_0):= J_Q(U;x_0) + \phi(U;x_0).
\end{equation}

To obtain a sampling-based representation of \eqref{eq:hard_unconstrained}, we employ a variational formulation over the space of probability distributions.
Let $\kappa$ and $\kappa_0$ be probability distributions over $\mathbb{R}^{Nm}$, where $\kappa_0$ serves as a reference distribution.
Provided that $\kappa$ is absolutely continuous with respect to $\kappa_0$, the Kullback-Leibler (KL) divergence of $\kappa$ from $\kappa_0$ is defined as
\begin{equation}
D(\kappa\|\kappa_0) := \int \log\left(\frac{d\kappa}{d\kappa_0}(U)\right)\,\kappa(dU).
\end{equation}

\begin{lem}[Variational formula \cite{boue1998variational}]
For any $\lambda > 0$,
    \begin{equation}\label{eq:variational_formula}
        \min_{\kappa} \left\{\mathbb{E}_{U\sim\kappa}[\mathcal J(U; x_0)] + \lambda D(\kappa\|\kappa_0)\right\} = -\lambda \log{Z},
    \end{equation}
    where
    $Z:=\mathbb{E}_{U\sim\kappa_0}[\exp(-\mathcal J(U;x_0)/\lambda)]$ 
    and the minimizer is
    \[
        \kappa^{*}(U) = \frac{\kappa_0(U)\exp(-\mathcal J(U; x_0) / \lambda)}{Z}.
    \]
    Moreover, the corresponding point estimate is given by
    \begin{equation}\label{eq:U_hat_def}
        \hat U(x_0) := \mathbb{E}_{U\sim\kappa^*}[U] = \frac{\mathbb{E}_{U\sim\kappa_0}\!\left[U \exp(-\mathcal J(U; x_0) / \lambda)\right]}{\mathbb{E}_{U\sim\kappa_0}\!\left[\exp(-\mathcal J(U; x_0) / \lambda)\right]}.
        \end{equation}
\end{lem}

By drawing $K\in\bbN$ i.i.d. samples $U^{(i)}\sim \kappa_0$ for $i=1,\ldots,K$, we can approximate \eqref{eq:U_hat_def} via Monte Carlo:
\begin{equation}\label{eq:U_hat}
    \hat U_K(x_0) = \frac{\sum_{i=1}^{K} U^{(i)} \exp(-\mathcal J(U^{(i)}; x_0) / \lambda)}{\sum_{i=1}^{K} \exp(-\mathcal J(U^{(i)}; x_0) / \lambda)}.
\end{equation}
By the law of large numbers, $\hat U_K(x_0) \to \hat U(x_0)$ as $K \to \infty$, and $\kappa^*$ clearly concentrates on $\arg\min \mathcal J(U;x_0)$ as $\lambda \to 0$.

\subsection{Exponential Tilting and Efficient Encrypted Sampling}\label{subsec:reference_dist}
The estimator \eqref{eq:U_hat} involves evaluating the extended-value cost $\mathcal{J}(U; x_0)$ for $K$ trajectory samples.
Homomorphically computing this trajectory cost online is expensive because evaluating the quadratic term $U^\top H U$, for example, requires at least $Nm$ calls of $\otimes$ for encrypted matrix-vector multiplication with the diagonal-based method of~\cite[Section~4.3]{HaleShou14}.

In what follows, we propose a method to circumvent the computation of this expensive quadratic term. 
This is achieved by applying a closed-form transformation to the reference sampling distribution, allowing it to naturally absorb the quadratic term.
Consequently, the resulting estimator can be evaluated without costly encrypted matrix-vector multiplications, leading to a highly efficient online protocol.

To facilitate the closed-form transformation, we set the reference distribution $\kappa_0$ to be a zero-mean Gaussian with covariance $\Sigma_0 \succ 0$
\begin{equation}
    \kappa_0=\mathcal N(0,\Sigma_0).
\end{equation}
The zero-mean assumption is adopted on $\kappa_0$ for simplicity but the same derivation extends to a general Gaussian $\mathcal N(m_0,\Sigma_0)$.

\begin{thm} \label{thm:tilt}
    The exponentially tilted distribution defined by $\tilde \kappa(dU) \propto \exp(-J_Q(U;x_0)/\lambda)\kappa_0(dU)$ is a Gaussian distribution $\tilde\kappa = \mathcal N\left(m_U(x_0),\,\Sigma_U\right)$ with
    \begin{subequations}\label{eq:thm1}
     \begin{align}
        m_U(x_0) &:= -\frac{1}{\lambda}\,\Sigma_U S^\top x_0, \\
        \Sigma_U &:=\left(\Sigma_0^{-1}+\frac{1}{\lambda}H\right)^{-1}.\label{eq:sigmaU}
    \end{align}
    \end{subequations}
\end{thm}
\begin{proof}
Multiplying $\exp(-J_Q(U;x_0)/\lambda)$ with the Gaussian density $\kappa_0=\mathcal N(0,\Sigma_0)$ yields an exponent of the form
\begin{align*}
   &-\frac{J_Q(U;x_0)}{\lambda} - \frac{1}{2}U^\top \Sigma_0^{-1}U\\
   &= 
-\frac12 U^\top \left(\Sigma_0^{-1}+\frac{1}{\lambda}H\right) U
-\frac{1}{\lambda}(S^\top x_0)^\top U. 
\end{align*}
Completing the square for $U$ reveals that this exponent corresponds to $\mathcal N(m_U(x_0), \Sigma_U)$, with all terms independent of $U$ absorbed into the normalization constant.
This concludes the proof.
\end{proof}

Note that the exponential weight appearing in \eqref{eq:U_hat_def} can be decomposed as
\begin{equation}
\exp\left(-\tfrac{\mathcal J(U;x_0)}{\lambda}\right) = \exp\left(-\tfrac{J_Q(U;x_0)}{\lambda}\right) r(U;x_0),
\end{equation}
where the feasibility weight $r(U;x_0) := \exp(-\phi(U;x_0)/\lambda)$ equals one when $U\in\mathcal F(x_0)$ and zero otherwise.

By applying the change of measure from $\kappa_0$ to $\tilde{\kappa}$, it follows from Theorem~\ref{thm:tilt} that the quadratic cost term can be naturally absorbed into the tilted distribution, thus simplifying \eqref{eq:U_hat_def} as
\begin{align*}
    \hat U(x_0) =\frac{\mathbb{E}_{U\sim\tilde{\kappa}}\left[U  r(U;x_0) \right]}{\mathbb{E}_{U\sim\tilde{\kappa}}\left[r(U;x_0)\right]}.
\end{align*}
As a result, the estimator \eqref{eq:U_hat} can be rewritten as 
\begin{equation}\label{eq:U_hat_mc_tilted}
    \hat U_K(x_0) = \frac{\sum_{i=1}^{K} U^{(i)} r(U^{(i)};x_0)}{\sum_{i=1}^{K} r(U^{(i)}; x_0)},
\end{equation}
where $U^{(i)}\sim\tilde{\kappa}$ for $i=1,\ldots, K$. Since $H$ and $S$ are known, \eqref{eq:thm1} can be computed explicitly and sampling from $\tilde{\kappa}$ is easy.

\subsection{Polynomial Surrogate of Feasibility}\label{subsec:poly_approx}

The feasibility weight $r(U; x_0)$ in \eqref{eq:U_hat_mc_tilted} is a hard indicator of constraint satisfaction. 
Therefore, $r(U; x_0)$ cannot be directly evaluated under HE because it requires comparison and branching, i.e., determining the sign of the constraint residual. Moreover, it is discontinuous and can not be represented by a finite degree polynomial.
For these reasons, we replace this hard indicator with a polynomial surrogate that can be evaluated using encrypted arithmetic.

\subsubsection{Aggregate violation score}

Recall that an input trajectory $U\in\bbR^{Nm}$ is feasible if and only if the constraint residual $g(U; x_0)$ satisfies $g_j(U;x_0)\le 0$ for all components $j=1, \ldots, p$.
Motivated by this characterization, we quantify the constraint violation via the following aggregate score
\begin{equation}\label{eq:violation_measure}
    s(U;x_0) := \textstyle\sum_{j=1}^{p} \bigl[g_j(U;x_0)\bigr]_+.
\end{equation}
Feasible samples yield $s(U;x_0)=0$, while infeasible samples incur a positive score proportional to the total violation.

\subsubsection{Chebyshev polynomial approximation of $[\,\cdot\,]_+$}
Let $h_\ell: \bbR\to\bbR$ denote a degree-$\ell$ polynomial of the form $h_\ell(\gamma) := \textstyle\sum_{k=0}^{\ell} c_k\,\gamma^k$ obtained via Chebyshev approximation \cite{KhanMich23} of $[\cdot]_+$ on the interval $[-B_\ell,B_\ell]$ for some $B_\ell > 0$.
Then, there exists a uniform bound $\delta_\ell>0$ such that
\begin{align}\label{eq:unifBound}
    \left | [\gamma]_+ - h_\ell(\gamma) \right| \le \delta_\ell, ~~~~ \forall \gamma \in [-B_\ell,B_\ell].
\end{align}

With the polynomial $h_\ell$, we have the surrogate score
\begin{equation}\label{eq:violation_poly}
    s_\ell(U;x_0) := \textstyle\sum_{j=1}^{p} h_\ell \left(g_j(U;x_0)\right).
\end{equation}
Let $\mathcal{D}_{B_\ell}(x_0) := \{U : \|g(U;x_0)\|_\infty \le B_\ell\}$ be the approximation domain.
Since $\tilde\kappa(\cdot\mid x_0)$ is Gaussian, the bound $B_\ell$ can be chosen such that $\|g(U^{(i)};x_0)\|_\infty \le B_\ell$ with high probability.
The following lemma is immediate from \eqref{eq:unifBound}.

\begin{lem}\label{lem:poly_error}
    For all $U\in\mathcal{D}_{B_\ell}(x_0)$, 
    \begin{align*}
        |s(U;x_0) - s_\ell(U;x_0)| \le p\,\delta_\ell.
    \end{align*}
\end{lem}

\subsubsection{Threshold correction and final estimator}

Small approximation errors in $s_\ell$ can accidentally penalize feasible samples.
To correct this, we introduce the thresholded score
\begin{equation}\label{eq:s_bar}
    \bar{s}_\ell(U;x_0) := \max\bigl\{s_\ell(U;x_0) - \tau_s,\;0\bigr\},
\end{equation}
where $\tau_s\ge 0$ is a tunable parameter.
The corresponding feasibility weight is defined as
\begin{equation}\label{eq:rp_def}
    \bar{r}_\ell(U;x_0) := \exp\bigl(-\eta\,\bar{s}_\ell(U;x_0)\bigr),
\end{equation}
where $\eta>0$ is chosen to be sufficiently large, so that $\bar{r}_\ell(U;x_0)$ closely approximates $r(U;x_0)$.

The following corollary states that an appropriate choice of $\tau_s$ guarantees that feasible samples are not penalized.
\begin{cor}\label{cor:threshold}
    If
    $\tau_s \ge p\,\delta_\ell$ then $\bar{r}_\ell(U;x_0) = 1$
    for every feasible $U \in \mathcal{D}_{B_\ell}(x_0)$.
\end{cor}

Replacing $r(U^{(i)};x_0)$ in \eqref{eq:U_hat_mc_tilted} with $\bar{r}_\ell^{(i)} := \bar{r}_\ell(U^{(i)};x_0)$ gives the final estimator
\begin{equation}\label{eq:U_hat_mc_poly}
    \hat{U}_{K,\ell}(x_0) := \frac{\sum_{i=1}^K U^{(i)}\,\bar{r}_\ell^{(i)}}{\sum_{i=1}^K \bar{r}_\ell^{(i)}},
\end{equation}
whose numerator and denominator consists of polynomial operations, aside from thresholding and exponentiation.

\subsection{Variational Encrypted MPC Protocol}
Algorithms~\ref{alg:vempc_offline} and~\ref{alg:vempc_online} summarize the client–cloud protocol realizing the proposed VEMPC. The protocol separates computation into an offline and an online phase.

\subsubsection{Offline preprocessing}

The offline phase precomputes quantities that depend only on the MPC parameters and the reference sampling distribution. 
Let $L_U\in\bbR^{Nm\times Nm}$ denote the Cholesky factor of $\Sigma_U$ in \eqref{eq:sigmaU}, i.e.,
$\Sigma_U = L_U L_U^\top$. 
Then, for some $\xi^{(i)} \sim \mathcal{N}(0, I)$, any sample $U^{(i)}\!\sim \tilde{\kappa}$ admits the following affine representation
\begin{equation}\label{eq:Ui}
    U^{(i)} = m_U(x_0) + L_U \xi^{(i)}
\end{equation}
Similarly, the constraint residual can be written as
\begin{equation}
    g(U^{(i)};x_0) = b(x_0) + \Gamma \xi^{(i)},
\end{equation}
where $b(x_0) := G m_U(x_0) - h(x_0)\in\bbR^p$ and  $\Gamma := G L_U\in\bbR^{p\times Nm}$.

Since both $L_U$ and $\Gamma$ are independent of the state $x_0$, the client encrypts and transmits $\Enc(L_U)$ and $\Enc(\Gamma)$ to the cloud during the offline phase.
The cloud then generates Gaussian noise samples $\xi^{(i)} \sim \mathcal N(0,I)$, for $i=1,\ldots, K$, and precomputes the following encrypted quantities\footnote{With slight abuse of notation, $\Enc$ and $\otimes$ denote matrix encryption and homomorphic matrix-vector multiplication, respectively. Both can be implemented via the vector-level operations in Section~\ref{prelim}; see \cite[Section~4.3]{HaleShou14}.}
\begin{align} \label{eq:enc_pert}
    \bfc_{LU}^{(i)} &:= \Enc(L_U) \otimes \Enc(\xi^{(i)}),\\
    \bfc_\Gamma^{(i)} &:=  \Enc(\Gamma)\otimes \Enc(\xi^{(i)}). \nonumber
\end{align}

Since these quantities depend only on the system model and not on the current state, this preprocessing is performed once. Moreover, these precomputed values can be cached for use during the online phase.

\begin{algorithm}[t]
\caption{Offline Protocol}
\label{alg:vempc_offline}
\begin{algorithmic}[1]
\Require MPC parameters $(H,G)$, reference covariance $\Sigma_0 \succ 0$, temperature $\lambda > 0$, number of samples $K$
\Ensure Cached ciphertexts $\{\bfc_{LU}^{(i)},\, \bfc_\Gamma^{(i)}\}_{i=1}^K$ stored at cloud 
\renewcommand{\algorithmicrequire}{\# \textbf{Client (trusted)}}
\Require
\State Compute $\Sigma_U := (\Sigma_0^{-1} + \tfrac{1}{\lambda}H)^{-1}$
\State Factorize $\Sigma_U = L_U L_U^\top$ and compute $\Gamma := G L_U $
\State Transmit \{$\Enc(L_U),\Enc(\Gamma)$\} to the cloud

\renewcommand{\algorithmicrequire}{\# \textbf{Cloud (untrusted)}}
\Require
\For{$i = 1,\ldots,K$}
    \State Sample $\xi^{(i)}\sim \calN(0,I)$
    \State $\bfc_{LU}^{(i)} \leftarrow \mathsf{Enc}(L_U) \otimes \Enc(\xi^{(i)})$
    \State $\bfc_\Gamma^{(i)} \leftarrow  \Enc(\Gamma)\otimes \Enc(\xi^{(i)})$
\EndFor
\end{algorithmic}
\vspace{-2pt}
\end{algorithm}

\subsubsection{Online execution}
At each MPC update step $t$, with $x_0 := x(t)$, the client encrypts and transmits $\mathsf{Enc}(m_U(x_0))$ and $\mathsf{Enc}(b(x_0))$ to the cloud.
Then, for each $i=1,\ldots,K$, the cloud retrieves the cached quantities $\bfc_{LU}^{(i)}$ and $\bfc_\Gamma^{(i)}$ to form the encrypted trajectory samples
\begin{equation}
    \bfU^{(i)} = \mathsf{Enc}(m_U(x_0)) \oplus \bfc_{LU}^{(i)},
\end{equation}
together with the encrypted constraint residuals
\begin{equation}
    \bfg(U^{(i)};x_0) = \mathsf{Enc}(b(x_0)) \oplus \bfc_\Gamma^{(i)}.
\end{equation}

To evaluate $s_\ell(U^{(i)};x_0)$ homomorphically, the cloud first applies the polynomial surrogate $h_\ell$ to $\bfg(U^{(i)};x_0)$, followed by a standard sum-reduction procedure \cite[Section~4.1]{HaleShou14} that iteratively rotates and adds the ciphertext to compute the aggregated score, which we denote by $\bfs_\ell(U^{(i)};x_0)$.

The cloud returns $\{\bfU^{(i)},\,\bfs_\ell(U^{(i)};x_0) \}_{i=1}^{K}$ back to the client. 
After decryption, the client computes the thresholded scores and corresponding desirability weights to construct the estimator $\tilde U_{K, \ell}(x_0)$, as in Lines~\ref{line:start}--\ref{line:end} of Algorithm~\ref{alg:vempc_online}. 
Finally, the first control input is applied to the system.

\begin{algorithm}[t]
\caption{Online Protocol}
\label{alg:vempc_online}
\begin{algorithmic}[1]
\Require Current state $x(t)$, threshold $\tau_s>0$, and $\eta>0$
\Ensure Control input $u(t)$

\renewcommand{\algorithmicrequire}{\# \textbf{Client (trusted)}}
\Require
\State $x_0 \gets x(t)$
\State Compute $m_U(x_0) = -\frac{1}{\lambda}\Sigma_U S^\top x_0$
\Statex and $b(x_0)=Gm_U(x_0)-h(x_0)$
\State Transmit \{$\Enc(m_U(x_0)),\Enc(b(x_0))$\} to the cloud

\renewcommand{\algorithmicrequire}{\# \textbf{Cloud (untrusted)}}
\Require

\For{$i = 1,\ldots,K$}
\State $\bfU^{(i)} \leftarrow \mathsf{Enc}(m_U(x_0)) \oplus \bfc_{LU}^{(i)}$
\State $\bfg(U^{(i)};x_0) \leftarrow \mathsf{Enc}(b(x_0)) \oplus \bfc_\Gamma^{(i)}$ 
\State $\bfs_\ell(U^{(i)};x_0) \leftarrow \sum_{j=1}^p \RotCt_{j-1}(h_\ell(\bfg(U^{(i)};x_0)))$
\EndFor
\State Return $\{\bfU^{(i)},\, \bfs_\ell(U^{(i)};x_0)\}_{i=1}^K$ to the client

\renewcommand{\algorithmicrequire}{\# \textbf{Client (trusted)}}
\Require
\For{$i = 1,\ldots,K$}\label{line:start}
    \State $\tilde{U}^{(i)}\gets \Dec(\bfU^{(i)})$
    \State $\check{s}_\ell^{(i)} \gets [\Dec(\bfs_\ell(U^{(i)};x_0))]_1$ \, ($[\cdot]_1$, first element)
    \State $\tilde{s}_\ell(U^{(i)};x_0) \gets \max\left(\check{s}_\ell(U^{(i)};x_0) - \tau_s,\, 0\right)$
    \State Compute weight $\tilde{r}_\ell^{(i)} = \exp(-\eta\, \tilde{s}_\ell(U^{(i)};x_0))$
\EndFor
\State Compute $\tilde{U}_{K,\ell}(x_0) := \dfrac{\sum_{i=1}^{K} \tilde{U}^{(i)} \tilde{r}_\ell^{(i)}}{\sum_{i=1}^{K} \tilde{r}_\ell^{(i)}}$ \label{line:end}
\State $u(t)\gets [I_m, 0]\tilde{U}_{K,\ell}(x_0)$
\end{algorithmic}
\end{algorithm}

\subsubsection{Encryption Error}
Since CKKS supports approximate arithmetic, the proposed protocol computes noisy versions of the trajectory samples and feasibility weights. The following result shows that encrypted computation within the protocol produces bounded errors propagated by homomorphic computations of the plaintext algorithm.
\begin{thm}\upshape\label{thm:errBound}
    For any fixed $x_0\in\bbR^n$ and $U^{(i)}\in\bbR^{Nm}$ defined as in \eqref{eq:Ui}, there exist $\sfB^U,\sfB^s \in \calO(N_\ckks/\Delta)$ such that
    \begin{align}
        \left\|\tilde{U}^{(i)} -  U^{(i)} \right\| &\le \sfB^U \label{eq:errU}\\
        \left\| \tilde{s}_\ell(U^{(i)};x_0) - \bar{s}_\ell(U^{(i)};x_0) \right\| & \le  \sfB^s . \label{eq:errs}
    \end{align}
\end{thm}
\begin{proof}
    First, we establish the bound in \eqref{eq:errU}. 
Using the method in \cite[Sec.~4.3]{HaleShou14}, evaluating $\bfc_{LU}^{(i)}$ requires performing $Nm$ homomorphic multiplications followed by $Nm$ homomorphic rotations. 
Consequently, the decryption error is bounded as
\begin{align}\label{eq:matVecErr}
\!\!\!\!\left\|  \Dec(\bfc_{LU}^{(i)}\!) \!-\! L_U\xi^{(i)} \right\| \!&=\! \left\| \Dec(\Enc(L_U) \!\otimes\! \Enc(\xi^{(i)})\!) \!-\! L_U\xi^{(i)} \right\| \nonumber\\
&\!\!\!\!\!\!\!\le Nm \left( \sfB^\otimes(\|L_U\|,\|\xi^{(i)}\|) + \sfB^\Rot \right)
\end{align}
by Lemma~\ref{lem:homo}.
Further applying Lemma~\ref{lem:homo} and \eqref{eq:matVecErr} yields
\begin{align*}
    &\left\| \tilde{U}^{(i)} \!-\! U^{(i)} \right\| \!=\! \left\| (\Dec(\Enc(m_U(x_0))) + \Dec(\bfc_{LU}^{(i)})) \!-\! U^{(i)}\right\| \\
    &\le \sfB^\Enc +  Nm\cdot \Big( \sfB^\otimes(\|L_U\|,\|\xi^{(i)}\|) + \sfB^\Rot \Big)=:\sfB^U.
\end{align*}
Since $(N,m,L_U,\xi^{(i)})$ are fixed, we have $\sfB^U\in\calO(N/\Delta)$.

For \eqref{eq:errs}, we omit dependence on $U^{(i)}$ and $x_0$ for brevity.
Using a similar argument to \eqref{eq:matVecErr}, $ \left\|\Dec(\bfg) - g \right\| \le \sfB^U$.
Combining this with Corollary~\ref{cor:poly} results in
\begin{align*}
    &\left\| \Dec(h_l(\bfg)) - h_l(g)\right\| \\
    & \le \left\| \Dec(h_l(\bfg)) - h_l(\Dec(\bfg))\right\| + \left\|  h_l(\Dec(\bfg)) - h_l(g)\right\| \\
    & \le \sfB^\poly(h_l,\bfg) + L_h B^U,
\end{align*}
where $L_h>0$ denotes the Lipschitz constant of $h_l$ on $[\left\| g\right\|_\infty - \sfB^U,\left\| g\right\|_\infty + \sfB^U]$.
Because the sum-reduction protocol iteratively applies homomorphic addition and rotation, it accumulates at most $p\sfB^\Rot$ error.
Therefore, defining $\sfB^s:=p\sfB^\Rot+ \sfB^\poly(h_l,\bfg) + L_h B^U$ satisfies the condition in \eqref{eq:errs}, and this concludes the proof.
\end{proof}
In particular, the error magnitude scales with the security parameter and  the CKKS precision parameter, by $\calO(N_\ckks/\Delta)$, implying the trade-off.

\subsection{Performance of the VEMPC Protocol}\label{subsubsec:performance}
The proposed protocol admits two complementary levels of parallelism: \emph{plaintext-level} parallelism across trajectory samples (from sampling-based algorithm) and \emph{ciphertext-level} parallelism enabled by Single-Instruction-Multiple-Data (SIMD) packing (from CKKS scheme).

\subsubsection{Plaintext-level}
The trajectory samples $U^{(i)}$ are independent and can be evaluated concurrently.
If $n_{\mathrm{workers}}\in\bbN$ number of processing units are available, the samples can be partitioned into batches so that each worker processes approximately $K_{\mathrm{worker}} = K/n_{\mathrm{workers}}$ samples.

\subsubsection{Ciphertext-level}\label{subsubsec:cipher}
The SIMD packing allows multiple plaintext values to be embedded into the slots of a single ciphertext, and processed simultaneously.

In the proposed protocol, a worker may process one ciphertext containing
$K_{\mathrm{worker}}$ trajectory samples packed as
\[
[U^{(1)};\cdots;U^{(K_{\mathrm{worker}})}], 
\quad \; U^{(i)} \in \mathbb{R}^{Nm},\;\; i=1,\ldots,K_{\mathrm{worker}},
\]
where each $U^{(i)}$ denotes the $i$-th control trajectory sample.
The corresponding residual vectors can also be packed similarly, allowing each worker to process them simultaneously with a similar computation time needed for a single sample.

The VEMPC protocol exploits these two levels of parallelism simultaneously. SIMD-packing enables multiple samples in a single ciphertext, while ciphertext batches can be processed concurrently by multiple workers. If $n_{\mathrm{workers}}$ workers process ciphertexts, each containing $K_{\mathrm{worker}}$ samples, the total number of trajectory samples evaluated at each MPC step is $K = n_{\mathrm{workers}}K_{\mathrm{worker}}$.

\section{Numerical Example}

We applied the proposed VEMPC protocol to an inverted pendulum system \cite{FranPowe86}.
By linearizing and discretizing the system with the sampling period of \SI{50}{\milli\s}, the model of the form \eqref{eq:sys} is obtained as
\begin{align*}
    A = 
    \begin{bmatrix}
        1.0246 & 0.0504 \\
        0.9890 & 1.0246
    \end{bmatrix}, ~~~ 
    B =
    \begin{bmatrix}
        0.0251 \\
        1.0082
    \end{bmatrix},
\end{align*}
where the parameters are set as $m=\SI{0.2}{\kg}$ (mass), $l=\SI{0.5}{\m}$ (length), and $g=\SI{9.81}{\m/\s^2}$ (gravitational acceleration). 
The state $x(t)=:[\theta(t);\dot{\theta}(t)]$ consists of the angular position $\theta(t)$ and angular velocity $\dot{\theta}(t)$ of the pendulum, and the input $u(t)$ represents the applied torque force. The initial state was set to $x(0)=[0.3;0.1]$.

For the MPC formulation, we set the prediction horizon as $N=10$ and the cost matrices as
\begin{align*}
    Q=
    \begin{bmatrix}
        50, & 0 \\
        0, & 5
    \end{bmatrix}, ~~~~ 
    Q_f = 2Q, ~~~~ R=0.1.
\end{align*}
The state and input constraints were defined as $\|\theta(t) \|\le 0.5$, $\|\dot{\theta}(t)\| \le 0.8$, and $\|u(t)\| \le 1$, resulting in a total of $p=60$ constraints.
For reference distribution and sampling, we set $\Sigma_0 = 0.25^2 I$, $\lambda=0.1$, and $K=240$.

Encryption parameters were set to $\{N_\ckks, \Delta\}=\{2^{13}, 2^{30}\}$, with the remaining parameters\footnote{Omitted here for brevity; refer to our implementations fully available at https://github.com/jsuh9/Variational-Encrypted-Model-Predictive-Control.} chosen to satisfy $128$-bit security \cite{AlbrChas21}. The Chebyshev polynomial order was set to $\ell=3$.
Implementing the parallelism from Section~\ref{subsubsec:performance} with $n_{\rm workers}=4$, the average online computation time for generating the actuator input at each time step was $\SI{28.662}{\milli\s}$ (mean over $T=40$ steps; $2$-second simulation shown in Fig.~\ref{fig:traj}), which falls within the sampling time. All simulation presented in this section was run on an Apple M5 laptop (4 Performance cores, \SI{24}{\giga\byte} unified memory).

The VEMPC protocol (Algorithms~\ref{alg:vempc_offline}--\ref{alg:vempc_online}) exhibits similar performance to the unencrypted variational MPC \eqref{eq:U_hat_mc_tilted}, effectively driving the state to the origin while satisfying the constraints, as shown in Fig.~\ref{fig:traj}.

\begin{table}[t]
  \centering
  \caption{Online time (mean $\pm$ std, \si{\milli\second}) for varying $\ell$, $N_{\ckks}$, or $K$. 
  }
  \label{tab:ablation}
  \setlength{\tabcolsep}{3pt}
  \renewcommand{\arraystretch}{1.0}
  \begin{tabular*}{\columnwidth}{@{\extracolsep{\fill}}
      c
      S[table-format=2.2(1.2), separate-uncertainty=true, table-align-uncertainty=true]
      S[table-format=2.2(1.2), separate-uncertainty=true, table-align-uncertainty=true]
      @{\hspace{6pt}}
      S[table-format=2.2(1.2), separate-uncertainty=true, table-align-uncertainty=true]
      S[table-format=2.2(1.2), separate-uncertainty=true, table-align-uncertainty=true]
    @{}}
    \toprule
    & \multicolumn{2}{c}{$\log N_{\ckks} = 13$}
    & \multicolumn{2}{c}{$\log N_{\ckks} = 14$} \\
    \cmidrule(lr){2-3}\cmidrule(lr){4-5}
    {$\ell$} & {$K=136$} & {$K=272$} & {$K=272$} & {$K=544$} \\
    \midrule
    3 & 28.47(1.76) & 30.65(2.31) & 60.39(1.92) & 62.66(4.50) \\
    4 & 33.61(1.25) & 32.73(3.34) & 69.47(4.07) & 71.91(5.99) \\
    5 & 37.59(5.55) & 35.09(1.36) & 75.27(4.18) & 76.82(2.59) \\
    \bottomrule
  \end{tabular*}
\end{table}

The dependence of the online computation time on protocol parameters such as polynomial degree $\ell$, ring dimension $N_{\rm ckks}$, and sample size $K$ is examined in Table~\ref{tab:ablation}. The computation time grows with $\ell$ because a higher degree requires more ciphertext multiplications, and with $N_{\rm ckks}$ because a larger ring dimension increases the size of each ciphertext. Ciphertext-level parallelism (Section~\ref{subsubsec:cipher}) is confirmed by the observation that computation times remain nearly constant for $K=136$ and $K=272$ under the same $\ell$ and $N_{\rm ckks}$, since the number of ciphertext multiplications and the ciphertext size are unchanged.



\begin{figure*}[!t]
    \centering
    \includegraphics[width=0.98\textwidth, height=30mm, keepaspectratio=false]{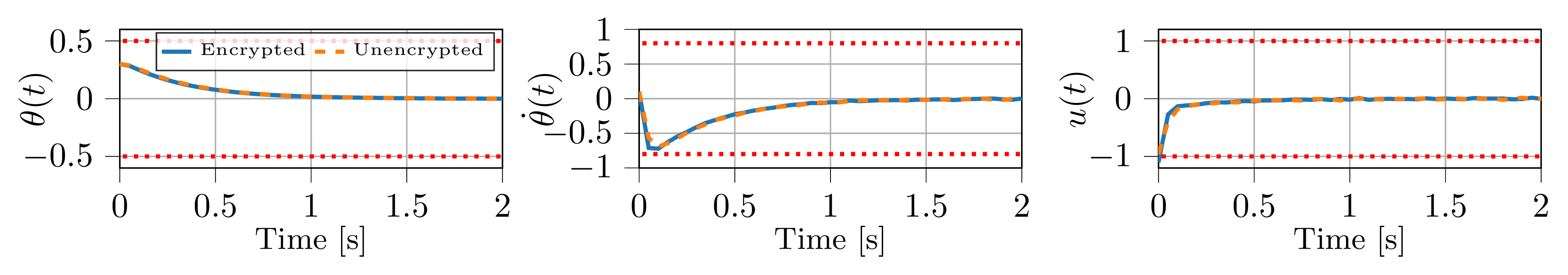}
    \vspace{-0.35cm}
    \caption{Comparison of unencrypted variational MPC (orange dashed) and VEMPC (blue solid) subject to constraints (red dotted). }
    \label{fig:traj}
\end{figure*}

\section{Conclusion} \label{final}
The difficulty of encrypted optimization (and hence encrypted MPC) stems from forcing solvers designed for cheap branching and iteration into an encrypted arithmetic where both are expensive.
A variational approach overcomes this incompatibility, since the resulting sample-based algorithm reduces encrypted computation to polynomial evaluations native to HE.
Exponential tilting eliminates the expensive quadratic cost from online computation, and together with two-level parallelism, the VEMPC protocol brings the control execution time to tens of milliseconds at $128$-bit security.

\addtolength{\textheight}{-4cm}




\bibliographystyle{IEEEtran}
\bibliography{IEEEabrv, vempc}

\end{document}